\documentclass[journal,twocolumn]{IEEEtran}
\usepackage{hyperref}
\usepackage{color,graphicx,amsmath,amssymb,amsthm,mathrsfs,cite}
\usepackage{tikz}
\usepackage{lipsum,adjustbox}
\usepackage{flushend}
\usepackage{cuted}
\usepackage{enumerate}
\usepackage{cool}
\usepackage{caption,subcaption}
\usetikzlibrary{arrows,snakes}
\usetikzlibrary{intersections}
\usetikzlibrary{plotmarks}
\usepackage{pgfplots}
\begin{document}
\tikzstyle{cnode}=[circle,draw]
\tikzstyle{cgnode}=[circle,draw]
\tikzstyle{cnode}=[circle,draw]
\tikzstyle{crnode}=[circle,draw]
\tikzstyle{conode}=[rectangle,draw]
\tikzstyle{cpnode}=[circle,draw]
\tikzstyle{rnode}=[rectangle,draw,outer sep=0pt]
\tikzstyle{prnode}=[rectangle,rounded corners,fill=blue!50,text width=4.5em,text centered,outer sep=0pt]
\tikzstyle{prnodebig}=[rectangle,rounded corners,fill=blue!50,text width=7em,text centered,outer sep=0pt]
\tikzstyle{prnodesimple}=[rectangle,draw,text width=4.5em,text centered,outer sep=0pt]
\tikzstyle{bigsnake}=[fill=green!50,snake=snake,segment amplitude=4mm, segment length=4mm, line after snake=1mm]
\tikzstyle{smallsnake}=[snake=snake,segment amplitude=0.7mm, segment length=4mm, line after snake=1mm]

\newtheorem{theorem}{Theorem}
\newtheorem{lemma}{Lemma}
\newtheorem{example}{Example}

\newcommand{\C}{\mathcal{C}}
\newcommand{\ch}{\text{ch}}

\title{Construction of Near-Capacity Protograph LDPC
  Code Sequences with Block-Error Thresholds}

\author{Asit~Kumar~Pradhan,~Andrew~Thangaraj~and~Arunkumar~Subramanian
\thanks{A. K. Pradhan and A. Thangaraj are with the Dept. of Electrical Engg.,
   IIT Madras, Chennai, India, Email: asit.pradhan,andrew@ee.iitm.ac.in.}
\thanks{A. Subramanian is with SanDisk Corporation, Milpitas, CA, USA.}}

\maketitle

\begin{abstract}
Density evolution for protograph Low-Density Parity-Check (LDPC) codes is
considered, and it is shown that the message-error rate falls
double-exponentially with iterations whenever the
degree-2 subgraph of the protograph is cycle-free and noise level is below threshold. Conditions for
stability of protograph density evolution are established and related
to the structure of the protograph. Using large-girth graphs, sequences
of protograph LDPC codes with block-error threshold equal to bit-error
threshold and block-error rate falling near-exponentially with
blocklength are constructed deterministically. Small-sized protographs are optimized to
obtain thresholds near capacity for binary erasure and binary-input Gaussian channels.   
\end{abstract} 

\section{Introduction}
Low-density parity-check (LDPC) codes, which are linear codes with
sparse parity-check matrices, are used today in several
digital communication system standards. Introduced by Gallager
\cite{gallagerthesis} in the 60s, the sparse parity-check matrices of modern LDPC codes are 
specified using the bit and check-node degree distributions of their Tanner graphs
\cite{mct}. The set of all Tanner graphs with a
given degree distribution defines an ensemble of LDPC codes.  

When decoded using the message-passing
algorithm over binary-input symmetric-output channels, the expected bit-error rate over the ensemble of LDPC
codes shows a \emph{threshold} phenomenon as
blocklength tends to infinity. There is a threshold channel parameter,
below which, the expected bit-error rate tends to fall rapidly for
large blocklength. The bit-error threshold, which is a function of the degree
distribution, is computed using a
procedure known as density evolution. The bit-error rate of a code in
the ensemble concentrates around the expected value; so, the threshold
is an important design parameter in practice. The practical design of LDPC codes
involves determining the degree distribution that maximizes the
threshold for a fixed rate. Given a degree distribution, a
parity-check matrix is sampled from the ensemble with several
heuristic criteria to simplify the complexity of implementation and for
acceptable performance \cite{ryan2009channel}.

The study of protograph LDPC codes, which are a special case of
Multi-Edge Type LDPC codes \cite{Richardson2004}, was initiated in \cite{Thrope}, and protograph LDPC
codes are the most popular codes today in theory (spatially-coupled
codes \cite{5695130,mitchell2014spatially}) and practice (included in WiFi and DVB-S2 standards). In
 \cite{5174517}, protographs are optimized for thresholds nearing
 capacity, and ensemble-averaged weight distribution is
 used to establish block-error threshold for protograph LDPC code
 ensembles. In \cite{5695100}, conditions on protograph for typical linear growth of
 minimum distance are derived. There have been numerous other work in the construction of
 protographs for several applications
 \cite{6253209,6266764,6905810,7005396,7045568}. Density evolution for
 protograph LDPC codes over the Binary Erasure Channel (BEC) was derived in
 \cite{LentmaierBECProtoDE}, and EXIT charts for protograph design
 were studied in \cite{Liva}.

While bit-error threshold is a popular design criterion, block-error
thresholds are important both from a theoretical and practical point
of view \cite{Lentermaier}. In spite of the importance, block-error
thresholds do not exist for many capacity-approaching degree
distributions that have degree-2 bit nodes. Another area of concern is random sampling in the construction
of LDPC and other modern codes. While concentration results are
useful, deterministic
constructions that have a provable block-error performance are the
ultimate goal of code design. Finally, the analytical properties of
protograph density evolution and
optimization of protographs using it are topics
that have not received much attention so far in the literature. This
work addresses the above shortfalls. 

The main contribution of this paper is the design and deterministic
construction of a sequence of large-girth, protograph LDPC codes with provable block-error
thresholds at rates approaching capacity. The idea of large-girth
 constructions, pioneered in Gallager's thesis \cite{gallagerthesis}, was studied in the
 context of block-error thresholds in
 \cite{Lentermaier} for the
 standard socket ensemble with minimum bit degree 3. In this work, the crucial
 property of double-exponential fall of message-error rate with
 iterations is extended to protograph LDPC ensembles that are allowed
 to contain degree-2 bit nodes under the condition that the degree-2
 subgraph of the protograph is cycle-free. The use of degree-2 bit
 nodes enables, through a carefully-designed differential evolution
 algorithm, the design of optimized protographs with block thresholds
 approaching capacity even at small sizes. To the best of our knowledge, the construction in this
work is perhaps the first deterministic LDPC code sequence with
guaranteed block-error rate behavior at rates close to capacity. As a specific example, we provide
a deterministic rate-1/2 protograph LDPC code sequence with a block-error threshold of 0.4953
over the BEC. 

The rest of this article is organized as follows. Section
\ref{sec:protograph-ldpc-code} introduces protograph LDPC codes and
their notation. The crucial property of double-exponential decay for
protograph density evolution and its stability are described in
Section \ref{sec:dens-evol-double}. The construction of large-girth
protograph LDPC codes is presented in Section
\ref{sec:large-girth-prot}. The optimization of protographs and
simulation results are given in Sections \ref{sec:optim-prot} and
\ref{sec:simulation-result}, respectively. Concluding remarks are made
in Section \ref{sec:conclusion}.
\section{Protograph LDPC Codes}
\label{sec:protograph-ldpc-code}
Following the notation in \cite{Thrope}, a protograph $G=(V\cup C,E)$ is a
bipartite graph with the bipartition $V$ and $C$ called the set of
variable or bit and
check nodes, respectively, and $E$ being the set of undirected edges
that connect a variable node in $V$ to a check node in $C$. Multiple
parallel edges are allowed between a variable node and a check
node. The nodes and edges in the protograph are ordered, and the $i$-th variable node, check node and edge in the protograph
are denoted, respectively, $v_i$, $c_i$ and $e_i$. The variable and
check nodes connected by an edge $e_i$ are denoted $v(e_i)$ and
$c(e_i)$, respectively.

A protograph can be represented by a base
matrix $B$ of dimension $|C|\times|V|$, whose $(i,j)$-th element $B(i,j)$ is the number of edges between $c_i$ and $v_j$. For example, consider a base matrix 
\begin{align}
B=\begin{bmatrix}
1&1&1&2\\
1&1&1&1
\end{bmatrix}.
\label{eq:23}
\end{align}
The protograph corresponding to the above base matrix is shown in Fig. \ref{fig:protex}.
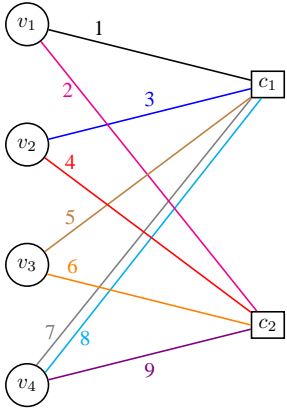
\begin{figure}[htb]
\centering
\begin{tikzpicture}[every node/.style={scale=0.8}]
\begin{scope}[node distance=2cm,>=angle 90,semithick]
\node[cnode] (v1) {$v_1$};
\node[crnode] (v2)[below of=v1] {$v_2$};
\node[cgnode] (v3)[below of=v2] {$v_3$};
\node[cpnode] (v4)[below of=v3] {$v_4$};
\node[conode] (c1)[right of=v2,xshift=2cm,yshift=1cm] {$c_1$}
  edge[black] node[near end,above]{1}(v1)
  edge[blue] node[above]{3} (v2);
\draw[gray] (v4.65) -- node[very near start,left]{7}(c1.220);
\draw[cyan] (v4.35) -- node[very near start,right]{8}(c1.245);
\node[conode] (c2)[right of=v3,xshift=2cm,yshift=-1cm] {$c_2$}
  edge[magenta] node[very near end,below]{2} (v1)
  edge[violet] node[midway,below]{9}(v4);
\draw[brown] (v3) --node[very near start,above]{5} (c1);
\draw[orange] (v3.335) --node[very near start,above]{6} (c2.180);
\draw[red] (v2) --node[very near start,above]{4}(c2);
\end{scope}    
\end{tikzpicture}
\caption{The protograph for the base matrix in \eqref{eq:23}.}
\label{fig:protex}
\end{figure}
The 9 different edges in this example are numbered as shown in the figure.

\subsection{Lifted graphs}
A copy and permute operation is applied to a protograph to obtain expanded or lifted graphs of different sizes \cite{Thrope}. A given protograph $G$ is copied, say $T$ times,
with the $t$-th copy having variable nodes denoted $(v,t)$, check
nodes denoted $(c,t)$, and edges denoted $(e,t)$ with $v\in V$, $c\in
C$ and $e\in E$. Then, for each edge $e$ in the protograph, we assign a permutation $\pi_e$ of the set $\{1,2,\ldots,T\}$. 
In the permute operation, an edge $(e,t)$ connecting $(v,t)$ and $(c,t)$ is permuted so as to connect variable node $(v,t)$ to check node $(c,\pi_e(t))$. An edge $(e_i,t)$ in the lifted graph is said to be of
type $e_i$ or, simply, type $i$. 

We will denote a lifted graph as $G(T,\Pi)=(V(T)\cup C(T),E(T,\Pi))$,
where $\Pi=\{\pi_e:e\in E\}$, or simply $G'=(V'\cup C',E')$ when the
exact $T$ and $\Pi$ are either clear or not critical. A lifted graph of a protograph can be thought of as a Tanner graph of an LDPC code, which is referred to as a protograph LDPC code.
The collection of these lifted graphs is called the
protograph ensemble of LDPC codes defined by $G$. Protograph LDPC codes are a 
special class of multi edge type (MET)-LDPC codes \cite{Richardson2004} with each edge in
the protograph being of a different type. The (designed) rate of the protograph LDPC code is
given by $1-|C|/|V|$.  The degree distribution of
check and variable nodes in the lifted graph is the same as that of
the protograph, but the protograph LDPC codes have a richer structure
than the standard ensemble when we consider computation graphs \cite{mct}. 

\subsection{Tree computation graphs}
Consider an edge of type $i$ in the lifted graph $G'$. The
$l$-iteration computation graph for the edge is defined as the subgraph of $G'$ obtained by traversing
down to depth $2l$ along all adjacent edges at the
variable node end \cite{mct}. An important observation is that the
vertex degrees and edge types in the computation graph are completely
determined by the protograph $G$. Further, let us suppose that the girth of $G'$ is greater than
$2l$, which makes the $l$-iteration computation graph a tree with no repeated
nodes. It is clear that the $l$-iteration tree computation graph for an edge
of a particular type in the protograph ensemble is
\emph{deterministic} and \emph{unique} in the sequence of vertex
degrees and edge types encountered. The
protograph $G$ completely determines the sequence of degrees and edge
types. So, in comparison with the standard socket ensemble \cite{mct},
no assumption on the distribution of tree computation
graphs is needed, and this makes density evolution analysis for large-girth protograph codes precise.

\section{Density Evolution and Double-Exponential Fall Property}
\label{sec:dens-evol-double}
A crucial fact that enables precise block-error rate guarantees from density evolution
is the property of double-exponential fall of message-error rate with
iterations \cite{Lentermaier}. For the standard socket ensemble, double-exponential fall
is possible only when the minimum degree is at least 3. In this section, we describe protograph density evolution
and show that double-exponential fall is possible even when degree-2
nodes are included in the protograph. We begin with the case of the binary erasure channel (BEC). 

\subsection{Binary erasure channel}
Let us consider the standard message-passing decoder \cite{mct} over a binary
erasure channel with erasure probability $\epsilon$, denoted BEC$(\epsilon)$, run on a lifted graph $G'$ derived
from a protograph $G=(V\cup C,E)$. Since the lifted graphs form an MET
ensemble with $|E|$ edge types, density evolution proceeds with $|E|$
erasure probabilities, one for each edge in the protograph
\cite{Richardson2004}. Let $x_t(i)$ be the probability that an erasure is sent from variable node
to check node along edge type $e_i$ in the $t$-th iteration. Similarly, let
$y_t(j)$ be the probability that an erasure is sent from check node to
variable node along edge type $e_j$ in the $t$-th iteration. The
protograph density evolution recursion \cite{LentmaierBECProtoDE} is given by
\begin{align}
\label{eq:30}x_0(i)&=\epsilon,\\
\label{eq:1}y_{t+1}(j)&=1-\prod_{i\in E_c(e_j)}(1-x_{t}(i)),\\
\label{eq:2}x_{t+1}(i)&=\epsilon\prod_{j\in E_v(e_i)}y_{t+1}(j),
\end{align}
for $t\ge0$ and $1\le i,j\le |E|$, where
$E_c(e)=\{i:c(e)=c(e_i),e\ne e_i\}$ and $E_v(e)=\{i:v(e)=v(e_i),e\ne e_i\}$ are the sets of other edge types
incident to the same check node and variable node, respectively, as
the edge $e$. The density evolution threshold, denoted
$\epsilon_{\text{th}}$, for the protograph-based LDPC code ensemble is  
defined as the supremum of the set of $\epsilon$ for which erasure
probability on each edge of the protograph tends to zero, as
$t\rightarrow \infty$,
i.e. $\epsilon_{\text{th}}=\sup\{\epsilon:\max_{i}x_t(i)\to0\}$. All
protographs in this work have minimum bit-node degree 2 ensuring that
$\epsilon_{\text{th}}$ is the threshold for variable node erasure
probability as well.

\subsubsection{Double-exponential fall}
\label{sec:asympt-behav-dens}
Consider a protograph $G$ with density evolution recursion as defined
in \eqref{eq:30}~-~\eqref{eq:2}. Because the recursion steps for
$x_{t+1}(i)$ and $y_{t+1}(j)$ involve
the neighbors of the edges $e_i$ and $e_j$, it is useful to visualize
\eqref{eq:30}~-~\eqref{eq:2} as iterative message passing on $G$ with
bit-to-check messages $x_t(i)$ and check-to-bit messages $y_t(j)$ in
iteration $t$. So, it is easy to see that all variable nodes in walks of
length $2t+1$ starting with $e_i$ are visited in the
computation of $x_t(i)$. Every time a variable node of degree at least
3 is traversed, multiplication of two or more terms occurs in the recursion as
per \eqref{eq:2} resulting in a squaring or higher power effect. Variable nodes of degree 2 result in a linear term
with no squaring. It turns out that ensuring at least a squaring effect at
regular intervals in every
walk is sufficient for double-exponential fall, and
this is made precise in the following theorem.
\begin{theorem}
Let a protograph $G$ be such that (1) there are no loops involving
only degree-2 variable nodes, and (2) every degree-2 variable node is
connected to a variable node of degree at least 3. Then, for
$\epsilon<\epsilon_{\text{th}}$, 
\begin{equation}
  \label{eq:31}
  x_t(i)=O(\exp(-\beta 2^{\alpha t}))
\end{equation}
for sufficiently large $t$, where $\alpha$, $\beta$ are positive constants.
\label{thm:double-expon-fall}
\end{theorem}
\begin{proof}
Let $\bar{x}_t=\max_i x_t(i)$ and let $|v_2|$ be the number of
degree-two variable nodes in $G$. We will show that there exists a
positive integer $R$ such that, for $t\ge R$,
\begin{align}
	\label{eq:6}\bar{x}_{t+|v_2|+1}\leq A (\bar{x}_t)^2,
\end{align}
where $A$ is a constant independent of $t$, and $A\bar{x}_R<1$. By
repeatedly applying \eqref{eq:6}, we
can readily show that 
\begin{equation}
  \label{eq:9}
\bar{x}_{R+i(|v_2|+1)}\leq A^{-1}(A\bar{x}_R)^{2^i},  
\end{equation}
for a positive integer $i$, which implies \eqref{eq:31}.

Suppose we have the upper bounds $x_{t+l}(i)\le C_l\bar{x}^{m(l,i)}_t$, where
$l\in\{0,1,\ldots,|v_2|\}$, $m(l,i)\in\{1,2\}$ and $C_l$ is a positive
constant independent of $t$. We will propagate the bounds through one round of \eqref{eq:1}~-~\eqref{eq:2} to obtain bounds $x_{t+l+1}(i)\le
C_{l+1}\bar{x}^{m(l+1,i)}_t$. For $l=0$, we have $C_0=1$ and
$m(0,i)=1$. We will show that $m(|v_2|+1,i)=2$ for all $i$, which
proves \eqref{eq:6}. 

We will use the following inequality. For any $x\in[0,1]$
and a positive integer $d$, 
\begin{align}
\label{eq:3}(d-1)x\geq1-(1-x)^{d-1}.
\end{align}
First consider the RHS of \eqref{eq:1} for a fixed $j$. Let 
\begin{equation}
  \label{eq:34}
n(l,j)=\min_{i\in E_c(e_j)}m(l,i).   
\end{equation}
Since $\bar{x}_t\le 1$, we have $x_{t+l}(i)\le C_l\bar{x}^{n(l,j)}_t$. So, 
\begin{align*}
	1 - x_{t+l}(i) &\geq 1 - C_l\bar{x}^{n(l,j)}_t, \quad i\in E_c(e_j),\\
	\Rightarrow \prod_{i \in E_c(e_j)} \left( 1 - x_{t+l}(i) \right) &\geq \left( 1 - C_l\bar{x}^{n(l,j)}_t \right)^{r(j)-1},
\end{align*}
where $r(j)$ is the degree of $c(e_j)$. Since $\bar{x}_t\to0$, for
large enough $t$, we have $C_l\bar{x}^{n(l,j)}_t<1$. So, using
\eqref{eq:1} and \eqref{eq:3}, we get, for large enough $t$, 
\begin{align}
	y_{t+l+1}(j) &\leq (r(j)-1) C_l\bar{x}^{n(l,j)}_t. \label{eq:4}
\end{align}
Now consider \eqref{eq:2} for a fixed $i$. We get
\begin{align}
	x_{t+l+1}(i) &= \epsilon \prod_{j \in E_v(e_i)} y_{t+l+1}(j)  \notag\\
	&\leq \epsilon ((r_{\max}-1)C_l)^{l(i) - 1}\prod_{j \in E_v(e_i)}\bar{x}^{n(l,j)}_t, \label{eq:5}\\
       &\leq \epsilon ((r_{\max}-1)C_l)^{l(i) -1}\bar{x}^{m(l+1,i)}_t,\label{eq:32}\\
&\le C_{l+1}\bar{x}^{m(l+1,i)}_t,
\end{align}
where $l(i)$ is the degree of $v(e_i)$, $r_{\max}$ is the maximum
check-node degree ($r_{\max}\ge2$), $C_{l+1}=\epsilon
\max_i((r_{\max}-1)C_l)^{l(i) -1}$ and we set
\begin{equation}
  \label{eq:33}
m(l+1,i)=\begin{cases}
1, \text{ if }\sum_{j \in E_v(e_i)}n(l,j)=1,\\
2, \text{ if }\sum_{j \in E_v(e_i)}n(l,j)\ge2.  
\end{cases}
\end{equation}
Note that $\sum_{j \in E_v(e_i)}n(l,j)=1$ only when $v(e_i)$ is a degree-2 variable node and $n(l,j)=1$ for the single
edge $e_j\in E_v(e_i)$. 

We now claim that $m(|v_2|+1,i)=2$. The proof for the claim is by
contradiction. Suppose that $m(|v_2|+1,i)=1$ for some $e_i$. Then, for
the single edge $e_j\in E_v(e_i)$, $n(|v_2|,j)=1$, which in turn
implies $m(|v_2|,i')=1$ for some $e_{i'}\in E_c(e_j)$. Proceeding in
this manner, there exists a walk in $G$ of length $|v_2|+1$ containing only degree-2 variable nodes. This is a contradiction
because $G$ has exactly $|v_2|$ degree-2 variable nodes, and by the
assumptions of the theorem, $G$ has no loops involving degree-2 variable
 nodes, and every degree-2 variable node in $G$ is connected to at least one
 variable node of degree at least 3. 
\end{proof}
Now, if there is a cycle involving degree-2 nodes in the protograph,
we can show, using a method similar to the proof above (after setting
$x_t(i)=0$ when $v(e_i)$ has degree at least 3), that 
$x_t(i)$ for an edge $e_i$ in the degree-2 cycle falls at most
exponentially with $t$. Therefore, the degree-2 subgraph being
cycle-free is a necessary and sufficient condition for
double-exponential fall of message error probability in protograph
density evolution. We remark that the condition of degree-2 subgraph
being loopfree has been used before in the context of typical linear growth of minimum
distance \cite{5695100}\cite{5454136}.
\subsubsection{large-girth lifted graph sequences and block-error threshold}
\label{sec:large-girth-lifted}
Consider a protograph $G=(V\cup C,E)$ satisfying the conditions of
Theorem \ref{thm:double-expon-fall} and a lifted graph
$G'=G(T,\Pi)$. Let $n=T|V|$ denote the blocklength of the LDPC code
defined by $G'$, and consider message-passing decoding over BEC$(\epsilon)$. When $G'$ has girth $g$, the probability of erasure on
an edge of type $i$ from bit node to check node in iteration $t$ is
exactly equal to $x_t(i)$ if $t\le g/2-1$. So, for $t\le g/2-1$, the probability of
erasure from bit to check on any edge is upper bounded by
$\bar{x}_t$, and by the union bound, the probability of block error, denoted $P_B(n)$, is upper bounded as
$P_B(n)=O(n\bar{x}_t)$. In Section \ref{sec:large-girth-prot}, for a given protograph $G$, we provide constructions of lifted
graphs with large girth or girth growing as $\Theta(\log n)$. So, for
large-girth lifted graphs, the girth can be increased arbitrarily by
increasing $n$, and we have
\begin{equation}
  \label{eq:35}
  P_B(n)=O(n\bar{x}_t)=O(n\exp(-\beta2^{\alpha t}))
\end{equation}
for $\epsilon<\epsilon_{\text{th}}$ and sufficiently large $n$ using Theorem
\ref{thm:double-expon-fall}. Now, setting $t=c\log n$, $c>0$, in
\eqref{eq:35}, we get
\begin{equation}
  \label{eq:36}
  P_B(n)=O(n\exp(-\beta n^{c\alpha})),
\end{equation}
for $\epsilon<\epsilon_{\text{th}}$. By noting that $\lim_{n
  \rightarrow \infty} n^k  n \exp(-\beta n^{c\alpha}) \rightarrow 0$,
we can say that the block-error
probability $P_B(n)$ falls faster than $1/n^k$ for a positive integer $k$. 

Therefore, for large-girth protograph LDPC code sequences, if the protograph
satisfies the conditions of Theorem \ref{thm:double-expon-fall},
block-error threshold is equal to the bit error threshold $\epsilon_{\text{th}}$.

\subsubsection{Stability of protograph density evolution}
Let $\mathbf{x}_t=[x_t(1)\ x_t(2)\ \cdots\ x_t(|E|)]$ denote the vector of
bit-to-check erasure probabilities in iteration $t$
as per the protograph density evolution of
\eqref{eq:30}~-~\eqref{eq:2} for a protograph $G$. The density evolution recursion can be
represented as $\mathbf{x}_{t+1}=f(\mathbf{x}_t,\epsilon)$, where the $i$-th coordinate
of the vector function $f$ is 
\begin{equation}
  \label{eq:37}
 f_i(\mathbf{x}_t,\epsilon)=\epsilon\prod_{j\in E_v(e_i)}\left(1-\prod_{i'\in E_c(e_j)}(1-x_t(i'))\right).
\end{equation}
The monotonicity of $f_i$ with $x_t(i')$ and
$\epsilon$ is easy to establish \cite{Richardson2004}. We concern
ourselves with the stability of the recursion. Approximating $f$ using Taylor series around origin, we get
\begin{equation}
 \mathbf{x}_{t+1}=\nabla_f \; \mathbf{x}_t + e_f(\mathbf{x}_t),
\end{equation}
where $\nabla_f$ is the $|E|\times |E|$
gradient matrix of $f$ with $(i,i')$-th element,
denoted $[\nabla_f]_{ii'}$,
defined as the partial derivative of $f_i(\mathbf{x}_t,\epsilon)$ with respect
to $x_t(i')$ evaluated at the origin $\mathbf{x}_t=\mathbf{0}$, and $e_f(\mathbf{x}_t)$ is a length-$|E|$ vector satisfying
$||e_f(\mathbf{x}_t)||^2=O(||\mathbf{x}_t||^2)$ ($||\cdot||$ denotes Euclidean norm).
Letting $l(i)$ denote the degree of bit node $v(e_i)$,
we readily see from \eqref{eq:37} that
\begin{equation}
  \label{eq:38}
  [\nabla_f]_{ii'}=\left.\pderiv{f_i(\mathbf{x}_t,\epsilon)}{x_t(i')}\right|_{\mathbf{x}_t=\mathbf{0}}=\begin{cases}
0,\text{ if }l(i)\ne2,\\
\epsilon,\text{ if }l(i)=2, i'\in E_c(e_j),
\end{cases}
\end{equation}
where, for the case $l(i)=2$, $e_i$ and $e_j$ are the two edges
connected to the degree-2 node $v(e_i)$.

For sufficiently small $\mathbf{x}_t$, the convergence of
$\mathbf{x}_{t+1}=f(\mathbf{x}_t,\epsilon)$ to $\mathbf{0}$ depends on the
eigenvalues of $\nabla_f$ being less than one \cite{horn1990matrix}. To study the
eigenvalues of $\nabla_f$, we use Perron-Frobenius theory on eigenvalues of
non-negative matrices following \cite{RothblumNonnegative06}. For this purpose, we introduce some notation
and definitions. 

A directed graph $D(\mathbf{A})$ is associated with a nonnegative $n\times n$
matrix $\mathbf{A}$. The vertex set of $D(\mathbf{A})$ is $\{1,2,\ldots,n\}$ with a
directed edge from $i$ to $j$ if and only if the $(i,j)$-th element of
$\mathbf{A}$, is nonzero. A directed graph $D$ is said to be
\emph{strongly connected}
if there is a directed path between any two vertices of $D$. A nonnegative
square matrix $\mathbf{A}$ is said to be \emph{irreducible} if $D(\mathbf{A})$ is
strongly connected. For a non-negative square matrix $\mathbf{A}$, there exists
a permutation matrix $\mathbf{P}$ such that 
\begin{equation}
  \label{eq:39}
\mathbf{P}\mathbf{A}\mathbf{P}^T=\begin{bmatrix}
\mathbf{A}_{11}&\mathbf{A}_{12}&\mathbf{A}_{13}&\cdots&\mathbf{A}_{1s}\\
\mathbf{0}&\mathbf{A}_{22}&\mathbf{A}_{23}&\cdots&\mathbf{A}_{2s}\\
\mathbf{0}&\mathbf{0}&\mathbf{A}_{33}&\cdots&\mathbf{A}_{2s}\\
\vdots&\vdots&\ddots&\ddots&\vdots\\
\mathbf{0}&\mathbf{0}&\cdots&\mathbf{0}&\mathbf{A}_{ss}
\end{bmatrix},  
\end{equation}
where $\mathbf{A}_{ii}$ is either a square irreducible matrix or a $1\times 1$
zero matrix. The block upper-triangular form of \eqref{eq:39} is called the \emph{Frobenius
  normal form} of $\mathbf{A}$. Note that $D(\mathbf{P}\mathbf{A}\mathbf{P}^T)$ is isomorphic to $D(\mathbf{A})$
with vertices permuted by $\mathbf{P}$, and the eigenvalues of $\mathbf{P}\mathbf{A}\mathbf{P}^T$ are the
same as that of $\mathbf{A}$. So, for the purposes of stability, we will assume that the gradient matrix $\nabla_f$ is in Frobenius
normal form with the diagonal blocks denoted as $\nabla_{ii}$, $1\le
i\le s_f$, where $s_f$ denotes the number of diagonal blocks. The subgraphs $D(\nabla_{ii})$ are called the strongly
connected components of $D(\nabla_f)$.

The next two lemmas connect edges and cycles in $D(\nabla_f)$ to the
structure of the protograph $G$.
\begin{lemma}
The directed graph $D(\nabla_f)$ has an edge from $i$ to $i'$ if and
only if $l(i)=2$ and $i'\in E_c(e_j)$, where $e_j$ is the single edge
in $E_v(e_i)$. This implies the following: (1) Vertex $i$ is in a
strongly-connected component of $D(\nabla_f)$ only if $l(i)=2$; (2) for edge $(i,i')$ in $D(\nabla_f)$, there exists a path
$(e_i,e_j,e_{i'})$ in the protograph with $l(i)=l(j)=2$.
\label{lem:gradedge}
\end{lemma}
\begin{proof}
The lemma is a restating of \eqref{eq:38}. Claim (1) follows because
an edge needs to originate out of a vertex in a strongly-connected
component. Claim (2) follows from \eqref{eq:38}.
\end{proof}

\begin{lemma}
There is a length-$l$ cycle in $D(\nabla_f)$ if and only if there
is a length-$2l$ cycle in the subgraph of the protograph induced by degree-2 bit nodes.
\label{lem:gradcycle}
\end{lemma}
\begin{proof}
Let $(e_{i_1},e_{i_2},\ldots,e_{i_l},e_{i_1})$ be a cycle in
$D(\nabla_f)$. This implies directed edges
$(e_{i_m},e_{i_{m+1}})$, $1\le m\le l-1$, and $(e_{i_l},e_{i_1})$ in
$D(\nabla_f)$. By Lemma \ref{lem:gradedge}, there are edges
$e_{j_1},e_{j_2},\ldots,e_{j_l}$ such that
$(e_{i1},e_{j_1},e_{i_2},e_{j_2},\ldots,e_{i_l},e_{j_l},e_{i_1})$ is a
cycle in the protograph and $l(i_m)=l(j_m)=2$ for $1\le m\le l$.
\end{proof}
Structural stability conditions on $G$ following from the above two lemmas are collected
in the next theorem.
\begin{theorem}
  Consider a protograph $G=(V\cup C,E)$ with gradient
  matrix $\nabla_f$, whose Frobenius normal form has diagonal blocks
  $\nabla_{ii}$ ($1\le i\le s_f$). Let $G_2$ denote the subgraph of
  $G$ induced by degree-two bit nodes. Let $E_2$ denote set of edges of
  $G$ incident on degree-two bit nodes. Protograph density evolution over
  BEC$(\epsilon)$ is stable in each of the following
  cases:
  \begin{enumerate}
  \item for all $\epsilon$, if $G_2$ is cycle-free.
  \item for $\epsilon<1$, if no two cycles of $G_2$ overlap in an edge.
  \item for $\epsilon<1/r_{\max}$, where $r_{\max}=\max_{e\in
      E_2}|E_c(e)\cap E_2|$.
  \end{enumerate}
\label{thm:stab}
\end{theorem}
\begin{proof}
1) If the subgraph of $G$ induced by degree-2 bit nodes is cycle-free, we
get, by Lemma \ref{lem:gradcycle}, that there are no cycles in
$D(\nabla_f)$. So, there are no strongly-connected subgraphs in
$D(\nabla_f)$, which implies that $\nabla_{ii}=0$ for all $i$ in the
Frobenius normal form of $\nabla_f$. Therefore, the eigenvalues of
$\nabla_f$ are all 0, implying stability for all $\epsilon$.

\noindent 2) If no two cycles of $G_2$ overlap in an edge, the cycles
of $D(\nabla_f)$ do not overlap in a vertex or an edge by Lemma
\ref{lem:gradcycle}. So, the strongly-connected components of $D(\nabla_f)$ are cycles. Since $D(\nabla_{ii})$ is a cycle, the eigenvalues of
    $\nabla_{ii}$ have absolute value equal to $\epsilon$ \cite{chung1997spectral}, implying
    stability for $\epsilon<1$.

\noindent 3) The result follows because the maximum eigenvalue of
    $\nabla_{ii}$ is upper bounded by its maximum row sum \cite{horn1990matrix}, and
    $\epsilon r_{\max}$ is an upper bound on the maximum row
    sum of the matrices $\nabla_{ii}$, $1\le i\le s_f$.
\end{proof}
From Theorems \ref{thm:double-expon-fall} and \ref{thm:stab}, the
degree-2 subgraph of the protograph being cycle-free emerges as an
important design condition. Next, we provide an example to illustrate
the stability conditions for protographs.
\begin{example}
Consider a protograph whose subgraph induced by degree-two bit nodes,
denoted $G_2$, is as shown in Fig. \ref{deg2subgraph}.
\begin{figure}[htb]
{
\label{degree2subgraph}
\begin{center}
\begin{tikzpicture}[every node/.style={scale=1}]
\node (G2) {\tikz{
\begin{scope}[node distance=2cm,>=angle 90,semithick]
\node[cnode] (v1) {};
\node[cnode] (v2)[below of =v1] {};
\node[cnode] (v3)[below of =v2] {};
\node[conode] (c1)[right of=v1,yshift=-1cm] {}
 edge node[above]{$e_1$} (v1)
 edge node[near end, above]{$e_5$} (v2)
 edge node[near end, above]{$e_3$} (v3);
\node[conode] (c2)[below of=c1] {}
 edge node[near end, above]{$e_4$} (v1)
 edge node[near end, above]{$e_2$} (v2)
 edge node[above]{$e_6$} (v3);
\node (G2l) [below of = v3, xshift=1cm,yshift=1.5cm] {$G_2$};
\end{scope}}
};
\node (D) [right of = G2,xshift=2.8cm]{\tikz{
\begin{scope}[node distance=2cm,>=angle 90,semithick]
\node[cnode] (t1) {$e_1$};
\node[cnode] (t2) [right of=t1,xshift=-1cm,yshift=-1cm] {$e_2$}
  edge[bend left,->] (t1)
  edge[bend right,<-] (t1);
\node[cnode] (t3) [below of=t2,yshift=0.586cm] {$e_3$}
  edge[bend left,->] (t2)
  edge[bend right,<-] (t2);
\node[cnode] (t4) [below of=t3,xshift=-1cm,yshift=1cm] {$e_4$}
  edge[bend left,->] (t3)
  edge[bend right,<-] (t3);
\node[cnode] (t5) [left of=t4,xshift=1cm,yshift=1cm] {$e_5$}
  edge[bend left,->] (t4)
  edge[bend right,<-] (t4);
\node[cnode] (t6) [left of=t1,xshift=1cm,yshift=-1cm] {$e_6$}
  edge[bend left,->] (t5)
  edge[bend right,<-] (t5)
  edge[bend left,->] (t1)
  edge[bend right,<-] (t1);
\node [below of = t4, yshift=1.2cm] {$D(\nabla_{11})$};
\end{scope}}
};    
\end{tikzpicture}
\caption{Illustration of stability for protograph density evolution.}
\label{deg2subgraph}
\end{center}
}
\end{figure}
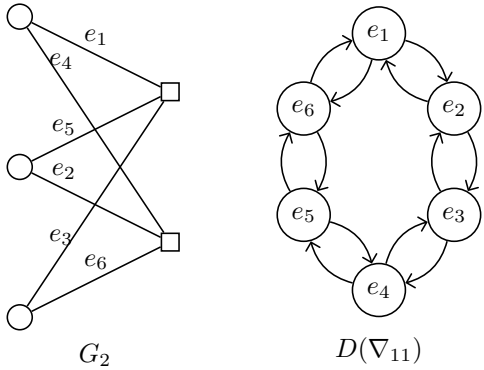
For such a protograph, the gradient graph has one non-trivial strongly-connected
component $D(\nabla_{11})$ as shown. Clearly, cases (1) and (2) of
Theorem \ref{thm:stab} do not apply. A quick calculation shows
$r_{\max}=2$, which results in stability for $\epsilon<0.5$. In this
case, an exact eigenvalue calculation matches with the bound based on $r_{\max}$.
\end{example}

\subsection{Binary-input symmetric channel}
\label{sec:AWGN}
The extension to binary-input symmetric channels uses the method of
Bhattacharya parameters, and we will be brief in our description
referring to \cite{mct} and \cite{Lentermaier} for details. A binary input channel $X\to Y$ with $X\in\{-1,+1\}$ is said to be
symmetric if the transition probability $p(Y|X)$ satisfies
$p(Y=y|X=+1)=p(Y=-y|X=-1)$. The standard message-passing decoder uses the
log-likelihood ratio (LLR) $l_0=\log\frac{p(y|+1)}{p(y|-1)}$ as input, and the
message passed from a bit node to a check node in iteration $t$ is
an LLR $l_t$ for the corresponding bit. The bit node operation is
simply addition, while the check node operation uses the standard $\tanh$
rule. Assuming that the all-$+1$s
codeword is transmitted and that the computation graph is a tree, the LLR
$l_t$ is of the form $\log\frac{p_t(y|+1)}{p_t(y|-1)}$, where $p_t(Y|X)$ is the
transition probability of a symmetric channel. Density evolution computes the
transition probability $p_t(Y|X=+1)$ using $p_{t-1}(Y|X=+1)$ and $p(Y|X)$.

For many symmetric channels of practical interest such as the
Binary-Input Additive White Gaussian Noise (BIAWGN) channel, the
transition probability $p(Y|X)$ is 
nonzero over the real line, which makes density evolution analysis cumbersome.
However, probability of message error in each iteration can be upper bounded by using the
Bhattacharyya parameter following \cite{Lentermaier}. The method in
\cite{Lentermaier} readily extends to protograph density evolution for a binary-input
symmetric channel as described next. 

The Bhattacharyya parameter for the channel corresponding to the bit-to-check message in the $t$-th iteration is defined as follows:
\begin{equation}
 B_t=\int_{-\infty}^{\infty}\sqrt{p_t(y|+1)p_t(y|-1)}dy.
\end{equation}
The probability of message error in the $t$-th
iteration is bounded above by the Bhattacharya parameter $B_t$.

Consider the standard message-passing decoder over a binary-input
symmetric channel $p(Y|X)$ run on a lifted graph $G'$ derived from a protograph
$G=(V\cup C,E)$. Protograph density evolution for this situation involves $|E|$ densities
$p^{(i)}_t(Y|X=+1)$, $1\le i\le|E|$, with corresponding Bhattacharya
parameters $B_t(i)$. The evolution of Bhattacharya parameters satisfies a
 set of inequalities given in the next lemma.
 \begin{lemma}
The Bhattacharyya parameters $B_t(i)$ satisfy
\begin{equation}
 B_{t+1}(i)\leq B_0\prod_{j\in E_v(e_i)}\sum_{i'\in
   E_c(e_j)}B_{t}(i')
\end{equation}
for $1\le i \le |E|$, where $B_0=\int_{-\infty}^{\infty}\sqrt{p(y|+1)p(y|-1)}dy$ is the
Bhattacharya parameter of the channel $p(Y|X)$, and $E_v$, $E_c$ are
as defined earlier.
\label{lem:bhatt}
 \end{lemma}
\begin{proof}
The proof follows the proof of Lemma 1 in \cite{Lentermaier} closely, and we
skip the details.
\end{proof}
Let $\ch(\sigma)$ be a family of binary-input symmetric output
channels, where $\sigma$ denotes the channel parameter with
$\ch(\sigma)$ being a degraded version of $\ch(\sigma')$ whenever $\sigma>\sigma'$. 
Let $\sigma_{\text{th}}$ be the threshold below
which the maximum probability of error in protograph density evolution
for a protograph $G$ tends to zero as $t\to\infty$. Since
probability of error tending to zero implies that Bhattacharya
parameter tends to zero, we have that, for $\sigma<\sigma_{\text{th}}$, the maximum Bhattacharya parameter $\max_iB_t(i)\to0$ as $t\to\infty$.

Now, using ideas similar to those used for the binary erasure channel, we
can show that the Bhattacharya parameter, and, hence, the
probability of message error, exhibits a double exponential fall with
iterations if the degree-2 subgraph of $G$ is cycle-free. This result
is stated as a theorem for reference.
\begin{theorem}
Let a protograph $G$ be such that (1) there are no loops involving
only degree-2 variable nodes, and (2) every degree-2 variable node is
connected to a variable node of degree at least 3. Then, for
$\sigma<\sigma_{\text{th}}$, 
\begin{equation}
  \label{eq:31a}
  B_t(i)=O(\exp(-\beta 2^{\alpha t}))
\end{equation}
for sufficiently large $t$, where $\alpha$, $\beta$ are positive constants.
\label{thm:double-expon-fall-bms}
\end{theorem}
\begin{proof}
The proof is similar to the proof of Theorem \ref{thm:double-expon-fall}.
\end{proof}
The statements about large-girth constructions in
Section \ref{sec:large-girth-lifted} for binary erasure channels carry over for
the binary-input symmetric channel case as well. In particular, a
sequence of large-girth protograph LDPC codes over binary-input
symmetric channels will have bit-error
threshold equal to block-error threshold, and block-error rate falling
near-exponentially with blocklength for noise levels below threshold,
whenever the degree-2 subgraph of the protograph is cycle-free.
\section{Large-Girth Protograph LDPC Codes}
\label{sec:large-girth-prot}
We have seen that a sequence (in $n$) of length-$n$ protograph LDPC codes with
girth increasing as $c\log n$, $c>0$, results in block-error rate
falling as $O\left(n\exp(-\beta n^{c\alpha})\right)$ (where
$\alpha,\beta>0$) below the message
or bit-error threshold of the protograph. 

The construction of large-girth regular graphs is a classic problem in
graph theory \cite{B.Bollobas}. For a recent construction and survey
of latest results, see
\cite{Dahan14}. For applications of large-girth graphs in the construction
of LDPC codes, see \cite{margulis82}, \cite{Rosenthal},
\cite{arunForensics}.  In this section, we show how
sequences of large-girth protograph LDPC codes can be constructed
starting from sequences of regular large-girth graphs. We also discuss
explicit deterministic constructions. Parts of this
construction were presented earlier in \cite{6620513}. 

\subsection{Construction of large-girth protograph LDPC codes}
Let $G=(V\cup C,E)$ be a protograph. The starting point for the
construction is a sequence of $|E|$-regular bipartite graphs $B_{n_i}=(V_{i}\cup C_{i},E_i)$,
$i=1,2,\ldots$, with $|V_i|=|C_i|=n_i$. The
existence of such sequences is well-known in graph theory, and we
provide explicit examples later on in this section. For now, we assume
that such a sequence is available.
\subsubsection{Edge coloring}
According to K\"onig's theorem \cite{bondy1976graph}, the graph $B_{n_i}=(V_{i}\cup C_{i},E_i)$ can be
edge-colored with $|E|$ colors numbered from 1 to $|E|$. We fix such a coloring. For a vertex
 $v\in V_{i}$, let $e_{j}(v)$, $j=1,2,\ldots,|E|$, denote the edge of
 color $j$ incident on $v$. Similarly, let $e_{j}(c)$ denote the edge
 of color $j$ incident on $c\in C_{i}$. 
\subsubsection{Node splitting}
Let us number the left vertices of the protograph $G$ as
$1,2,\ldots,|V|$, the right vertices as $1,2,\ldots,|C|$, and the
edges as $1,2,\ldots,|E|$. Let
$l(j)$ and $r(j)$ denote the left and right vertex indices of $G$ connected by the
edge $j$. 

From the graph $B_{n_i}$, we will construct a bipartite graph
$G'=(V'\cup C',E')$, where $|V'|=n_i|V|$, $|C'|=n_i|C|$ and
$|E'|=|E_i|=n_i|E|$, by operations that we call node splitting followed by
edge reconnecting. Every vertex $v\in V_{i}$ is split into 
$|V|$ vertices and denoted, say, as $v_1,v_2,\ldots,v_{|V|}\in V'$. Every vertex $c\in C_{i}$ is split into
$|C|$ vertices denoted $c_1,c_2,\ldots,c_{|C|}\in C'$. Now, we connect
the edge $e_j(v)$ originally incident on $v\in V_i$ to the new vertex
$v_{l(j)}\in V'$. Similarly, we connect the edge $e_j(c)$ incident on
$c\in C_i$ to the new vertex $c_{r(j)}\in C'$. 

The node splitting step is illustrated in Fig. \ref{fig:nodesplit}.
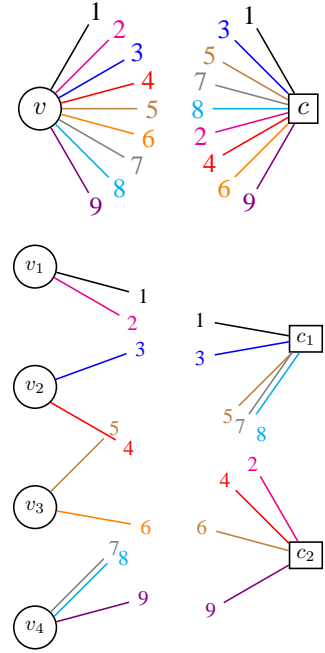
\begin{figure}[htb]
\begin{center}
\begin{tikzpicture}
\begin{scope}[node distance=2cm,>=angle 90,semithick]
\node[cnode] (v) {$v$};
\foreach \x/\c/\n in {60/black/1,45/magenta/2,30/blue/3,15/red/4,0/brown/5,345/orange/6,330/gray/7,315/cyan/8,300/violet/9}
{
\draw[\c] (v.\x) -- +(\x:1cm) +(\x:1.2cm) node{\n};
}
\node[conode] (chk) [right of=v,xshift=1.5cm]{$c$};
\foreach \x/\c/\n in {120/black/1,135/blue/3,150/brown/5,165/gray/7,180/cyan/8,195/magenta/2,210/red/4,225/orange/6,240/violet/9}
{
\draw[\c] (chk.\x) -- +(\x:1cm) +(\x:1.2cm) node{\n};
}
\node [below of=v,xshift=1.7cm,yshift=-2.5cm]{
\begin{tikzpicture}[every node/.style={scale=0.8}]
\begin{scope}[node distance=2cm,>=angle 90,semithick]
\node[cnode] (v1) {$v_1$};
\node[crnode] (v2)[below of=v1] {$v_2$};
\node[cgnode] (v3)[below of=v2] {$v_3$};
\node[cpnode] (v4)[below of=v3] {$v_4$};
\node[conode] (c1)[right of=v2,xshift=2.5cm,yshift=0.8cm] {$c_1$};
\node[conode] (c2)[right of=v3,xshift=2.5cm,yshift=-0.8cm] {$c_2$};
\draw[black] (v1.345) -- +(345:1cm) +(345:1.2cm) node{1};
\draw[magenta] (v1.330) -- +(330:1cm)  +(330:1.2cm) node{2};

\draw[blue] (v2.20) -- +(20:1cm)  +(20:1.2cm) node{3};
\draw[red] (v2.315) -- +(330:1cm)  +(330:1.2cm) node{4};

\draw[brown] (v3.45) -- +(45:1cm)  +(45:1.2cm) node{5};
\draw[orange] (v3.350) -- +(350:1cm)  +(350:1.2cm) node{6};

\draw[gray] (v4.45) -- +(45:1cm)  +(45:1.2cm) node{7};
\draw[cyan] (v4.30) -- +(45:1cm)  +(40:1.2cm) node{8};
\draw[violet] (v4.15) -- +(15:1cm)  +(15:1.2cm) node{9};

\draw[black] (c1.170) -- +(170:1cm)  +(170:1.2cm) node{1};
\draw[blue] (c1.190) -- +(190:1cm)  +(190:1.2cm) node{3};
\draw[brown] (c1.225) -- +(225:1cm)  +(225:1.2cm) node{5};
\draw[gray] (c1.225) -- +(235:1cm)  +(235:1.2cm) node{7};
\draw[cyan] (c1.245) -- +(235:1cm)  +(245:1.2cm) node{8};

\draw[magenta] (c2.120) -- +(120:1cm)  +(120:1.2cm) node{2};
\draw[red] (c2.145) -- +(135:1cm)  +(135:1.2cm) node{4};
\draw[brown] (c2.165) -- +(165:1cm)  +(165:1.2cm) node{6};
\draw[violet] (c2.210) -- +(210:1cm)  +(210:1.2cm) node{9};
\end{scope}    
\end{tikzpicture}
};
\end{scope}    
\end{tikzpicture}
\end{center}
\caption{Illustration of node splitting with the protograph of Fig. \ref{fig:protex}.}
\label{fig:nodesplit} 
\end{figure}
\subsubsection{Properties}
Two main properties are quite easy to prove. The first is that the
graph $G'$ obtained by node splitting is a lifted version of the
protograph $G$, and can be generated by a copy-permute operation on
 $G$. In the copy-permute operation, the protograph $G$ is copied $n_i$ times, and the
permutation of the edge type $j$ in $G$ is determined precisely by the
matching $M=\{e_j(v): v\in V_i\}$ in $B_{n_i}$. Numbering the
left/right vertices of $B_{n_i}$ from 1 to $n_i$, let $M$ map the
left vertex $t$ to the right vertex $M(t)$ in $B_{n_i}$. In the $t$-th copy of the
protograph $G$, the edge $(e,t)$ connecting $(v,t)$ to $(c,t)$ is
permuted to connect $(v,t)$ to $(c,M(t))$ in the lifted graph.

The second property is that the girth of $G'$ is at least as large as the girth
of $B_{n_i}$. This is easy to see because a cycle in $G'$ readily
maps to a cycle of the same length in $B_{n_i}$. 

In summary, we see that, given a sequence of $|E|$-regular bipartite
graphs $B_{n_i}$ with $n_i$ nodes in each bipartition and
girth at least $c\log n_i$, we can construct a sequence of liftings of
a protograph with $|E|$ edges such that the girth of the lifted graphs
grows at least as $c\log n_i$.   

\subsection{Deterministic constructions}
The construction method described above can
use any sequence of regular large-girth graphs. For completeness and
to give deterministic constructions, we describe the parameters of
two large-girth graph sequences called 
LPS graphs \cite{LPS} and $D(m,q)$ graphs \cite{Ustimenko}, which we have used in simulations.

\subsubsection{ LPS Graphs $X^{p,q}$} 
Let $p$ and $q$ be distinct, odd primes with $q>2\sqrt{p}$. The LPS
graph, denoted $X^{p,q}$ \cite{LPS}, is a connected, $(p+1)$-regular graph and has
the following properties:
\begin{itemize}
 \item If $p$ is a quadratic residue $\mod q$, then $X^{p,q}$ is a
   non-bipartite graph with $q(q^2-1)/2$ vertices and girth
   $g(X^{p,q}) \geq 2\log_p q$.
\item If $p$ is a quadratic non-residue $\mod q$, then $X^{p,q}$ is a
  bipartite graph with $q(q^2-1)$ vertices and girth
  $g(X^{p,q}) \geq 4\log_p q-\log_p 4$.
\end{itemize}
When $X^{p,q}$ is non-bipartite, we can convert it to a bipartite
graph using the following algorithm \cite{B.Bollobas}\cite{arunForensics}:
\begin{itemize}
 \item Given a graph $G$ with vertices $V(G)$ and edges $E(G)$, construct a copy
   $G^\prime$ with a new vertex set $V(G')$ and a new edge set $E(G')$. Let $f:V(G)\rightarrow
   V(G^\prime)$ be the 1-1 mapping from a vertex in $G$ to its copy in $G'$.
\item Create a bipartite graph $H$ with vertex set $V(G)\cup V(G^\prime)$ and
  edge set $E(H)=\{(x,f(y)):(x,y)\in E(G)\}$.
\end{itemize}
Following \cite{B.Bollobas}, it was shown in \cite{arunForensics} that $g(H)\geq g(G)$.
For constructing a sequence of $d$-regular large-girth graphs for an arbitrary $d$ using the LPS graphs, we use the following trick from \cite{arunForensics}. 
There exists an infinite number of primes $p$ such that $d$ divides $(p+1)$, i.e., $d|(p+1)$. For each such prime $p$ and a suitable $q$, we construct $X^{p,q}$ and split each $(p+1)$-degree node into 
$(p+1)/d$ nodes of degree $d$. As shown in \cite{arunForensics}, node splitting does not reduce girth and we have a large-girth graph of the required degree $d$.

\subsubsection{$D(m,q)$ graph}
The $D(m,q)$ graphs satisfy the following properties (\cite{Ustimenko}
and \cite{Viglione}):
\begin{enumerate}[(a)]
\item For a prime power $q$ and an integer $m\geq 2$, the girth of
  $D(m,q)$ satisfies
  \begin{equation}
    \label{eq:24}
g(D(m,q) \geq \begin{cases} m+5,\  m \text{ odd},  \\
 m+4,\ m \text{ even.}
\end{cases}      
  \end{equation}
\item For $q \geq 5$ and $2 \leq m \leq 5$, $D(m,q)$ is a connected bipartite graph with $2q^{m}$ vertices.
\item For $m \geq 6$, the graph $D(m,q)$ is disconnected. Because of
  edge transitivity all connected components are isomorphic. There are
  $q^{t-1}$ components of $D(m,q)$, where $t=\left\lfloor \frac{m+2}{4}\right\rfloor $.
Each component of $D(m,q)$ has $2q^{m-t+1}$ vertices and has girth
equal to $g(D(m,q))$ defined in \eqref{eq:24}.
Thus, for $m\ge 6 $, any connected component of
$D(m,q)$ can be used for constructing LDPC codes.
\end{enumerate}

\subsubsection{Comparison between $X^{p,q}$ and $D(m,q)$}
In the LPS construction $X^{p,q}$, to guarantee a minimum girth $g$, a
careful calculation shows that we must have
blocklength $n\sim p^{\frac{3g}{2}}$ or $n\sim p^{\frac{3g}{4}}$.

For $m \geq 6$, to guarantee girth $g$ in the $D(m,q)$ construction,
the blocklength grows as $n \sim q^{\frac{3g-13}{4}}$, which is smaller than that of $X^{p,q}$.
Hence, we can generate graphs of smaller block length by using the
$D(m,q)$ graph with the node-splitting algorithm. But unlike $X^{p,q}$, in $D(m,q)$, the vertex degree is always a power
of a prime, which implies that the number of edge types in the
protograph needs to be a prime power. The constructions in
\cite{Dahan14} work directly for an arbitrary degree, and could be
used as well.

\section{Optimization of Protographs}
\label{sec:optim-prot}
In this section, we describe the search procedure used for generating
optimized protographs. The erasure channel version was partly presented in
\cite{6620513}.
\subsection{Differential evolution}
We have optimized protographs using differential evolution \cite{DE}\cite{Storn}, where
we use the threshold given by density evolution as the cost function. The
salient steps of the differential evolution algorithm are described briefly in the following:
\begin{itemize}
 \item[1.] Initialization: For generation $G=0$, we randomly choose
   $N_P$ base matrices $B_{k,G}$, with $0\le k\le N_P-1$, of size $|C|\times|V|$, where
   $N_P=10|C||V|$. Each entry of $B_{k,G}$ is binary, chosen
   independently and uniformly. 
\item[2.] Mutation: Protographs of a particular generation are interpolated as follows.
\begin{align}
 M_{k,G}=[B_{r_1,G}+0.5(B_{r_2,G}-B_{r_3,G})],
\end{align}
where $r_1$, $r_2$, $r_3$ are randomly-chosen distinct values in the
range $[0,N_P-1]$, and $[x]$ denotes the absolute value of $x$ rounded to
the nearest integer.
\item[3.] Crossover: A candidate protograph $B'_{k,G}$ is chosen as
  follows. The $(i,j)$-th entry of $B'_{k,G}$ is set as the $(i,j)$-th
  entry of $M_{k,G}$ with probability $p_c$, or as the $(i,j)$-th
  entry of $B_{k,G}$ with probability $1-p_c$. We use $p_c=0.88$ in
  our optimization runs. In $B'_{k,G}$, if any cycle of degree-2 nodes
  emerges, edges are reassigned. 
\item[4.] Selection: For generation $G+1$, protographs are selected as
  follows. If the threshold of $B_{k,G}$ is greater than that of
  $B'_{k,G}$, set $B_{k,G+1}=B_{k,G}$; else, set $B_{k,G+1}=B'_{k,G}$.
\item[5.] Termination: Steps 2--4 are run for several generations (we
  run up to $G=6000$) and the protograph that gives the best threshold is chosen as the optimized protograph.
\end{itemize}
In the crossover step, we ensure that the subgraph induced by the degree-2
nodes of the protograph is a tree. This ensures that the block-error
threshold equals the bit-error threshold. If this condition is not
enforced, better thresholds might result from the optimization, but
with no guarantee of a block-error threshold.

The value of $p_c$ is the crossover step has been taken as 0.88 based
on trial and error. The optimization can be run with other values of
$p_c$, but we have obtained acceptable results with this value.
\subsection{Optimized protographs for BEC}
\label{subsec:BEC_opt_prot}
A few optimized protographs obtained from the above optimization
process are as follows. An optimized $3\times 12$, rate-3/4 protograph
with threshold 0.238 is given by the following base matrix:
 \begin{equation}
  \label{th=0.238}
\begin{bmatrix}
  1 & 1 & 0 & 0 & 7 & 4 & 1 & 1 & 0 & 0 & 0 & 0\\
 1 & 2 & 3 & 0 & 7 & 1 & 0 & 0 & 3 & 1 & 0 & 0\\
 1 & 5 & 5 & 3 & 4 & 0 & 1 & 2 & 0 & 1 & 3 & 3
\end{bmatrix}
\end{equation}
 An optimized $4\times 12$, rate-2/3 protograph
with threshold 0.32 is given by the following base matrix:
 \begin{equation}
  \label{th=0.32}
\begin{bmatrix}
 1 & 1 & 1 & 5 & 3 & 1 & 0 & 2 & 3 & 1 & 1 & 1\\
 0 & 1 & 0 & 6 & 0 & 0 & 0 & 2 & 0 & 1 & 1 & 1\\
 0 & 0 & 2 & 6 & 0 & 0 & 1 & 1 & 0 & 1 & 0 & 0\\
 2 & 0 & 1 & 2 & 0 & 1 & 2 & 4 & 0 & 4 & 1 & 1
\end{bmatrix}
\end{equation}
An optimized $4\times 8$, rate-1/2 protograph
with threshold 0.479 is given by the following base matrix:
\begin{equation}
  \label{eq:8}
\begin{bmatrix}
1 & 2 & 2 & 3 & 4 & 1 & 1 & 0  \\
0 & 1 & 0 & 0 & 5 & 0 & 0 & 1 \\
1 & 0 & 0 & 0 & 3 & 0 & 4 & 1\\
1 & 0 & 1 & 0 & 6 & 1 & 0 & 0 
\end{bmatrix}
\end{equation}
We observe that high thresholds are obtained even with small-sized
protographs. As the size increases, the thresholds get close to
capacity bounds.

An optimized $8\times 16$ protograph with threshold 0.486 is given by the following base matrix:
\begin{equation}
  \label{eq:7}
\begin{bmatrix}
1 & 2 & 0 & 0 & 1 & 0 & 0 & 4 & 0 & 0 & 0 & 0 & 0 & 0 & 0 & 1\\
0 & 1 & 0 & 0 & 0 & 1 & 0 & 0 & 2 & 2 & 1 & 0 & 0 & 0 & 1 & 1\\
0 & 3 & 1 & 2 & 1 & 0 & 0 & 0 & 4 & 0 & 0 & 3 & 2 & 2 & 0 & 3\\
0 & 5 & 0 & 0 & 0 & 0 & 1 & 1 & 0 & 0 & 1 & 0 & 0 & 1 & 0 & 0\\
1 & 3 & 1 & 1 & 1 & 2 & 0 & 0 & 1 & 0 & 0 & 0 & 0 & 0 & 0 & 0\\
1 & 5 & 0 & 0 & 0 & 3 & 1 & 0 & 0 & 0 & 1 & 0 & 0 & 0 & 0 & 0\\
0 & 4 & 0 & 0 & 0 & 0 & 0 & 1 & 1 & 0 & 0 & 0 & 0 & 0 & 0 & 1\\
0 & 5 & 0 & 0 & 0 & 0 & 0 & 0 & 0 & 1 & 0 & 0 & 1 & 0 & 1 & 0
\end{bmatrix}  
\end{equation}
A $16\times32$ protograph with threshold 0.4952 is given in \eqref{th=0.495}.
\begin{figure*}
\begin{equation}
\label{th=0.495}
\left[
\begin{array}{*{32}c}
 3 & 1 & 0 & 0 & 0 & 0 & 0 & 0 & 0 & 0 & 0 & 0 & 0 & 0 & 0 & 2 & 1 & 0 & 1 & 0 & 0 & 0 & 0 & 0 & 0 & 1 & 0 & 0 & 0 & 0 & 0 & 0\\
 4 & 0 & 2 & 0 & 0 & 0 & 0 & 0 & 0 & 1 & 0 & 1 & 0 & 0 & 0 & 0 & 0 & 0 & 0 & 0 & 0 & 0 & 0 & 0 & 0 & 0 & 0 & 0 & 0 & 0 & 0 & 0\\
 1 & 0 & 0 & 0 & 1 & 1 & 1 & 1 & 0 & 0 & 0 & 0 & 1 & 1 & 0 & 1 & 0 & 0 & 1 & 0 & 0 & 0 & 0 & 0 & 1 & 0 & 0 & 0 & 0 & 1 & 2 & 0\\
 4 & 0 & 0 & 0 & 1 & 1 & 0 & 1 & 0 & 0 & 0 & 0 & 0 & 0 & 0 & 0 & 0 & 0 & 0 & 0 & 0 & 0 & 0 & 0 & 0 & 0 & 0 & 0 & 0 & 0 & 1 & 0\\
 2 & 1 & 1 & 1 & 1 & 0 & 0 & 0 & 0 & 0 & 0 & 0 & 0 & 1 & 0 & 0 & 0 & 2 & 0 & 0 & 0 & 0 & 0 & 0 & 0 & 1 & 0 & 0 & 0 & 0 & 0 & 0\\
 2 & 0 & 0 & 0 & 0 & 1 & 1 & 2 & 1 & 0 & 0 & 0 & 1 & 0 & 0 & 0 & 0 & 0 & 1 & 1 & 0 & 0 & 0 & 0 & 0 & 0 & 0 & 1 & 1 & 1 & 0 & 0\\
 4 & 2 & 0 & 0 & 0 & 0 & 0 & 0 & 0 & 0 & 0 & 0 & 0 & 0 & 0 & 0 & 0 & 0 & 0 & 0 & 0 & 0 & 0 & 0 & 0 & 0 & 1 & 0 & 0 & 0 & 0 & 0\\
 4 & 1 & 1 & 1 & 0 & 0 & 0 & 0 & 0 & 0 & 2 & 0 & 0 & 1 & 0 & 0 & 0 & 0 & 0 & 0 & 0 & 0 & 0 & 1 & 0 & 0 & 0 & 0 & 0 & 0 & 0 & 0\\
 4 & 0 & 1 & 0 & 1 & 0 & 1 & 0 & 0 & 0 & 0 & 0 & 0 & 0 & 0 & 1 & 1 & 0 & 0 & 0 & 0 & 1 & 1 & 0 & 0 & 0 & 0 & 0 & 1 & 0 & 0 & 0\\
 4 & 0 & 0 & 0 & 1 & 0 & 0 & 2 & 0 & 1 & 0 & 0 & 1 & 0 & 0 & 0 & 0 & 0 & 0 & 0 & 0 & 0 & 0 & 1 & 0 & 2 & 0 & 0 & 0 & 0 & 0 & 1\\
 0 & 0 & 0 & 0 & 0 & 0 & 0 & 1 & 0 & 0 & 0 & 0 & 0 & 3 & 0 & 0 & 0 & 0 & 1 & 0 & 1 & 0 & 0 & 0 & 0 & 1 & 0 & 0 & 0 & 0 & 1 & 0\\
 2 & 0 & 0 & 0 & 0 & 1 & 0 & 0 & 0 & 0 & 0 & 1 & 2 & 1 & 1 & 0 & 0 & 0 & 0 & 0 & 0 & 0 & 0 & 0 & 0 & 0 & 0 & 1 & 0 & 1 & 0 & 0\\
 1 & 0 & 1 & 0 & 1 & 1 & 0 & 1 & 1 & 0 & 0 & 0 & 2 & 0 & 2 & 0 & 0 & 0 & 0 & 1 & 1 & 0 & 1 & 0 & 1 & 1 & 0 & 0 & 1 & 0 & 3 & 1\\
 4 & 0 & 0 & 0 & 0 & 0 & 0 & 0 & 0 & 0 & 0 & 1 & 0 & 2 & 0 & 0 & 0 & 0 & 2 & 0 & 0 & 0 & 0 & 0 & 0 & 0 & 0 & 0 & 0 & 0 & 0 & 0\\
 4 & 0 & 1 & 0 & 0 & 0 & 0 & 0 & 0 & 0 & 0 & 1 & 1 & 1 & 0 & 0 & 0 & 1 & 1 & 0 & 0 & 0 & 0 & 0 & 0 & 0 & 0 & 0 & 0 & 1 & 0 & 0\\
 3 & 0 & 2 & 0 & 0 & 0 & 0 & 0 & 0 & 0 & 1 & 0 & 0 & 1 & 0 & 2 & 0 & 0 & 0 & 0 & 0 & 1 & 0 & 0 & 0 & 0 & 1 & 0 & 0 & 0 & 0 & 0
\end{array}
\right]\end{equation}
\end{figure*}
The above protographs from our optimization runs are compared against other protographs in Table \ref{tab:optprot_BEC}.
We see that the optimized protographs give better thresholds than
irregular standard ensemble codes with minimum degree 3
\cite{arunForensics} and other construction such as AR4JA
\cite[Figure 7]{4036046}, and standards such as WIMAX \cite{6272299}
and DVB-S2 \cite{dvbs2}.
\begin{table}[htb]
  \centering
  \begin{tabular}{|c|c|c|c|c|}
    \hline
    {\bf Code type}&{\bf Rate}&{\bf Size}&{\bf Threshold}&{\bf Gap}\\
   \hline
   WIMAX & 0.5 & $12 \times 24$& 0.448&0.052\\
   \hline
   DVB-S2 & 0.444 & $25 \times 45$& 0.516&0.040\\
    \hline
    Standard ($l_{\min}=3$)&0.5&Not applicable&0.461&0.039\\
   \hline
   AR4JA & 0.5 & $4 \times 8$& 0.468&0.032\\
    \hline
     protograph in \eqref{eq:8}&0.5&$4\times 8$&0.479&0.021\\
    \hline
     protograph in \eqref{eq:7}&0.5&$8\times 16$&0.486&0.014\\
    \hline
   protograph in \eqref{th=0.495}&0.5&$16 \times 32$&0.4953&0.0047\\
   \hline
   \hline
   AR4JA & 0.67 & $2 \times 6$& 0.291&0.039\\
   \hline
   WIMAX & 0.67 & $8 \times 24$& 0.292&0.038\\
   \hline
   DVB-S2 & 0.67 & $15 \times 45$& 0.305&0.028\\
  \hline
   protograph in \eqref{th=0.32}&0.67&$4\times 12$&0.32&0.01\\
   \hline
   \hline
   WIMAX & 0.75 & $6 \times 24$& 0.212&0.038\\
   \hline
   DVB-S2 & 0.73 & $12 \times 45$& 0.232&0.018\\
   \hline
   protograph in \eqref{th=0.238}&0.75&$3 \times 12$&0.238&0.012\\
    \hline
  \end{tabular}
  \caption{Comparison of protograph thresholds for BEC.}
  \label{tab:optprot_BEC}
\end{table}
\subsection{Optimized protographs for BIAWGN channel}
For BIAWGN channel, the threshold of protograph density evolution is
computed using the EXIT chart method described in \cite{Liva}. A few
optimized protographs are given below, and their SNR thresholds
(denoted $\text{SNR}_{\text{th}}$) are
compared against the capacity-achieving SNR (denoted
$\text{SNR}_{\text{cap}}$) and other protographs such as 
AR4JA \cite[Figure 7]{4036046}, and those from the DVB-S2 and WIMAX standards in Table \ref{tab:optprot_AWGN}.
\begin{equation}
  \label{th=1.18}
\begin{bmatrix}
 2 & 0 & 0 & 1 & 7 & 0 & 1 & 0 & 0 & 0 & 1 & 1\\
 0 & 0 & 1 & 1 & 7 & 0 & 0 & 1 & 2 & 0 & 0 & 1\\
 4 & 1 & 1 & 1 & 5 & 0 & 1 & 0 & 0 & 0 & 1 & 3\\
 5 & 1 & 1 & 6 & 1 & 3 & 1 & 1 & 1 & 3 & 0 & 1
\end{bmatrix}
\end{equation}
\begin{equation}
  \label{th=1.79}
\begin{bmatrix}
 0 & 0 & 7 & 0 & 2 & 1 & 0 & 1 & 0 & 2 & 3 & 0\\
 2 & 3 & 7 & 2 & 2 & 3 & 1 & 3 & 2 & 3 & 5 & 3\\
 1 & 0 & 8 & 1 & 0 & 2 & 1 & 1 & 1 & 4 & 0 & 0
\end{bmatrix}
\end{equation}
\begin{figure*}
\begin{equation}
\left[
\label{th=0.3}
\begin{array}{*{32}c}
 0 & 1 & 0 & 0 & 1 & 0 & 2 & 1 & 0 & 0 & 0 & 0 & 2 & 0 & 2 & 0 & 1 & 1 & 0 & 2 & 1 & 0 & 0 & 0 & 2 & 0 & 2 & 0 & 0 & 0 & 0 & 0\\
 0 & 0 & 0 & 1 & 0 & 0 & 0 & 0 & 0 & 0 & 0 & 0 & 0 & 0 & 0 & 1 & 0 & 4 & 0 & 1 & 0 & 0 & 0 & 0 & 0 & 1 & 0 & 0 & 0 & 0 & 0 & 0\\
 0 & 0 & 1 & 2 & 1 & 0 & 0 & 0 & 2 & 0 & 0 & 0 & 0 & 0 & 0 & 0 & 1 & 0 & 0 & 0 & 0 & 0 & 0 & 0 & 0 & 0 & 0 & 0 & 0 & 0 & 0 & 0\\
 0 & 0 & 0 & 1 & 1 & 0 & 2 & 0 & 1 & 0 & 0 & 0 & 1 & 0 & 1 & 0 & 0 & 0 & 0 & 1 & 0 & 0 & 0 & 0 & 1 & 0 & 0 & 0 & 0 & 0 & 2 & 2\\
 1 & 0 & 1 & 0 & 0 & 0 & 0 & 0 & 0 & 0 & 1 & 0 & 1 & 0 & 0 & 0 & 1 & 3 & 0 & 1 & 1 & 0 & 0 & 0 & 0 & 0 & 0 & 0 & 0 & 0 & 0 & 0\\
 0 & 0 & 1 & 1 & 0 & 0 & 0 & 0 & 1 & 2 & 0 & 0 & 0 & 0 & 0 & 0 & 0 & 3 & 1 & 0 & 0 & 1 & 0 & 0 & 0 & 1 & 0 & 0 & 1 & 1 & 1 & 0\\
 1 & 0 & 0 & 0 & 0 & 0 & 0 & 0 & 0 & 0 & 0 & 3 & 1 & 0 & 0 & 0 & 0 & 1 & 0 & 0 & 0 & 2 & 1 & 0 & 0 & 1 & 0 & 0 & 0 & 0 & 0 & 0\\
 0 & 0 & 2 & 0 & 0 & 0 & 0 & 1 & 0 & 0 & 0 & 0 & 0 & 0 & 0 & 0 & 1 & 4 & 0 & 0 & 0 & 0 & 0 & 0 & 0 & 0 & 0 & 0 & 0 & 0 & 0 & 0\\
 0 & 0 & 2 & 0 & 0 & 0 & 0 & 1 & 0 & 1 & 1 & 0 & 0 & 1 & 0 & 0 & 2 & 0 & 0 & 0 & 0 & 0 & 0 & 1 & 0 & 0 & 1 & 2 & 0 & 0 & 2 & 1\\
 0 & 0 & 2 & 1 & 0 & 1 & 0 & 1 & 0 & 0 & 0 & 0 & 0 & 1 & 0 & 1 & 0 & 2 & 0 & 0 & 0 & 0 & 0 & 0 & 0 & 0 & 0 & 0 & 0 & 0 & 0 & 0\\
 0 & 0 & 1 & 1 & 0 & 0 & 2 & 0 & 0 & 0 & 0 & 0 & 0 & 0 & 0 & 0 & 0 & 4 & 0 & 0 & 0 & 0 & 0 & 0 & 0 & 0 & 0 & 0 & 0 & 0 & 0 & 1\\
 0 & 0 & 0 & 0 & 0 & 0 & 0 & 1 & 0 & 0 & 0 & 0 & 1 & 0 & 0 & 0 & 0 & 4 & 0 & 1 & 0 & 0 & 0 & 2 & 0 & 0 & 0 & 0 & 0 & 0 & 0 & 1\\
 0 & 0 & 0 & 0 & 1 & 0 & 0 & 0 & 1 & 0 & 0 & 0 & 0 & 0 & 1 & 0 & 0 & 1 & 2 & 1 & 0 & 0 & 1 & 0 & 1 & 0 & 0 & 0 & 2 & 1 & 0 & 0\\
 0 & 0 & 2 & 0 & 0 & 1 & 2 & 0 & 0 & 0 & 0 & 0 & 0 & 0 & 2 & 0 & 0 & 1 & 0 & 0 & 0 & 0 & 0 & 1 & 0 & 0 & 0 & 0 & 0 & 0 & 0 & 0\\
 1 & 0 & 1 & 2 & 2 & 0 & 0 & 0 & 0 & 0 & 0 & 0 & 0 & 0 & 0 & 1 & 0 & 4 & 0 & 0 & 0 & 0 & 0 & 0 & 0 & 0 & 0 & 0 & 0 & 0 & 0 & 0\\
 0 & 1 & 0 & 0 & 0 & 0 & 0 & 1 & 1 & 0 & 0 & 0 & 1 & 0 & 0 & 0 & 0 & 1 & 2 & 0 & 0 & 0 & 0 & 0 & 0 & 0 & 0 & 1 & 0 & 0 & 0 & 2
\end{array}
\right]\end{equation}
\end{figure*}

\begin{table}[htb]
  \centering
  \begin{tabular}{|c|c|c|c|c|}
    \hline
    {\bf Code type}&{\bf Rate}&{\bf Size}&{$\text{SNR}_{\text{th}}$ (dB)}&{Gap (dB)}\\
   \hline
   DVB-S2 & 0.444 & $25 \times 45$& 0.474&1.042\\
   \hline
   WIMAX & 0.5 & $12 \times 24$& 0.812&0.625\\
   \hline
   AR4JA & 0.5 & $4 \times 8$& 0.496&0.309\\
    \hline
   protograph in \eqref{th=0.3}&0.5&$16 \times 32$&0.3&0.113\\
   \hline
\hline
   WIMAX & 0.67 & $8 \times 24$& 2.799&0.491\\
   \hline
   DVB-S2 & 0.67 & $15 \times 45$& 2.749&0.441\\
   \hline
   AR4JA & 0.67 & $2 \times 6$& 1.338&0.279\\
   \hline
   protograph in \eqref{th=1.18}&0.67&$4\times 12$&2.429&0.121\\
   \hline
\hline
   DVB-S2 & 0.73 & $12 \times 45$& 3.62&0.498\\
   \hline
   WIMAX & 0.75 & $6 \times 24$& 3.83&0.443\\
   \hline
   protograph in \eqref{th=1.79}&0.75&$3 \times 12$&3.551&0.164\\
    \hline
  \end{tabular}
  \caption{Comparison of protograph thresholds for BIAWGN channel.}
  \label{tab:optprot_AWGN}
\end{table}
Note that the optimized protographs presented in this section satisfy the conditions of Theorem
\ref{thm:double-expon-fall} and have block-error threshold same as the
bit-error threshold. Also, the block-error rate falls inverse-polynomially (or better) in blocklength under the large-girth construction as described in
Sections \ref{sec:large-girth-lifted} and
\ref{sec:large-girth-prot}. Moreover, even for small sizes of
protographs such as $4\times 8$, $8\times 16$ or $16\times 32$, the
optimization results in thresholds that are near capacity.

\section{Simulation Results}
\label{sec:simulation-result}
We now present simulation results that confirm the predicted threshold
behavior for both the BEC and the BIAWGN channel.  
\subsection{BEC}
Protographs in \ref{subsec:BEC_opt_prot} can be lifted using
$D(m,q)$ graphs ($m$: positive integer, $q$: prime power) from Section \ref{sec:large-girth-prot} since they have a prime number of
edges. The parameters of the constructed LDPC codes are as given in
Table \ref{tab:parameter_BEC}.
 \begin{table}[htb]
  \centering
  \begin{tabular}{|c|c|c|c|c|}
    \hline
    {\bf Code type}&{\bf Rate}&$\mathbf{m,q}$&{\bf Blocklength}\\
    \hline
     $16\times 32$ protograph in \eqref{th=0.495}&0.5&2, 173&957728\\
     \hline
     $4\times 12$ protograph in \eqref{th=0.32}&0.66&2, 61&44652\\
     \hline
     $3\times 12$ protograph in \eqref{th=0.238}&0.75&2, 61&44652\\
     \hline
  \end{tabular}
  \caption{Parameters of protograph LDPC codes used for simulation over BEC.}
  \label{tab:parameter_BEC}
\end{table}

The standard message passing decoder is simulated over the BEC, and the bit and block-error rate
curves are shown in Fig. \ref{fig:bec}. Error rates of AR4JA, DVB-S2
and WIMAX codes of the same rate are shown for comparison. These codes
were lifted to blocklengths comparable to those in Table
\ref{tab:parameter_BEC} of the same rate. The rate-1/2 code was lifted to a
length of 90000, while the rate-2/3 and rate-3/4 codes were lifted to
a length of around 45000.
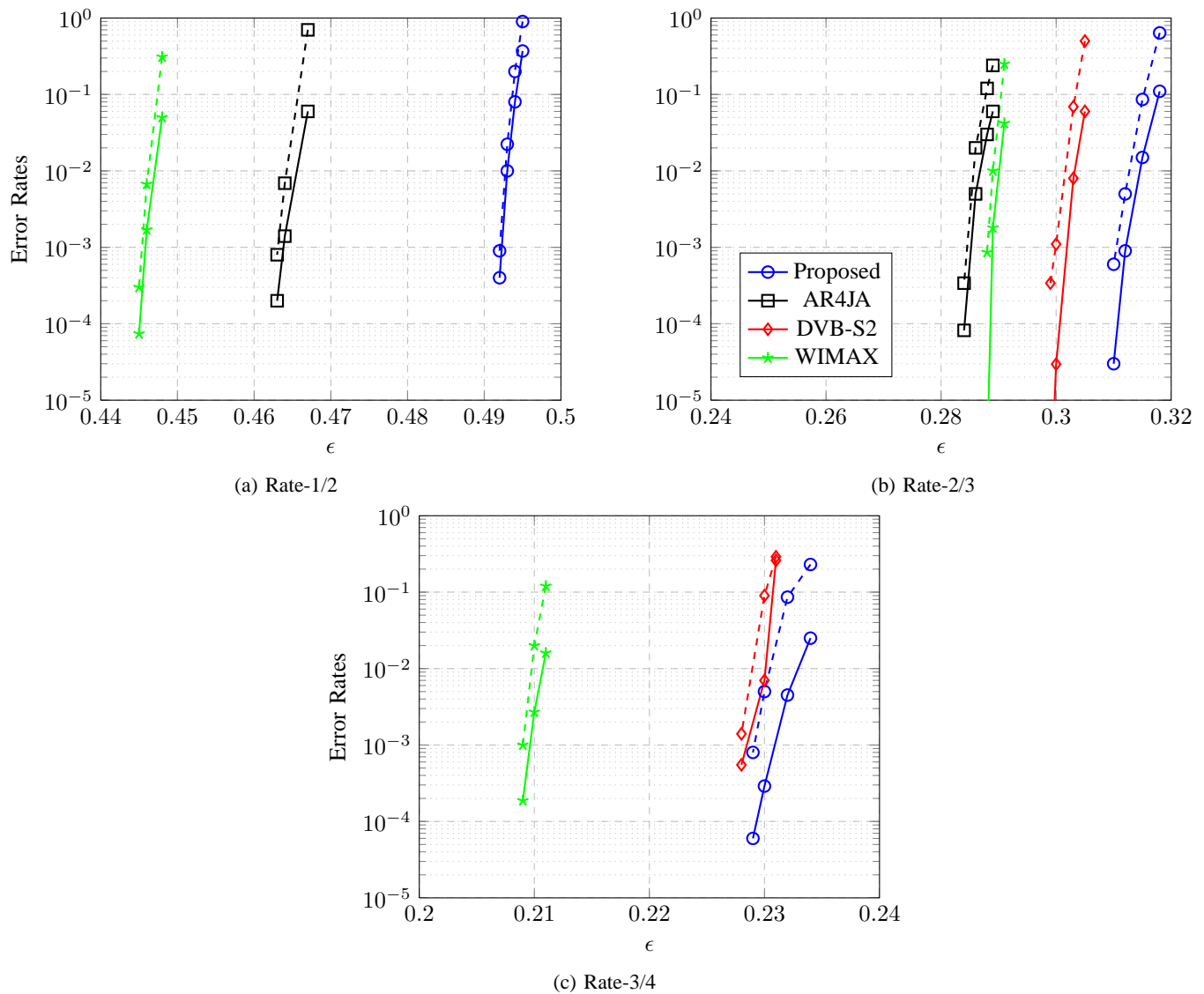
\begin{figure*}
\centering
\begin{subfigure}[b]{0.47\textwidth}
  \centering
   \begin{tikzpicture}
\begin{semilogyaxis}[width=\textwidth,xmin=0.44,xmax=0.5,ymin=1e-5,ymax=1,grid=both,minor grid style=dotted,major grid style=dashed,xlabel=$\epsilon$,ylabel=Error Rates,legend style={at={(axis cs:0.52,2e-5)},anchor=south west}]
\addplot [color=blue,solid,mark=o,mark size=2.5,mark options={solid},thick]
coordinates {
    (0.495,0.37)
    (0.494,0.08)
    (0.493,0.01)
    (0.492,0.0004)
};
\addplot [color=blue,dashed,mark=o,mark size=2.5,mark options={solid},thick]
coordinates {
    (0.495,0.9)
    (0.494,0.2)
    (0.493,0.0223)
    (0.492,0.0009)
};
\addplot [color=black,solid,mark=square,mark size=2.5,mark options={solid},thick]
coordinates {
    (0.467,6e-2)
    (0.464,1.4e-3)
    (0.463,2e-4)
};  
\addplot [color=black,dashed,mark=square,mark size=2.5,mark options={solid},thick]
coordinates {
    (0.467,7e-1)
    (0.464,6.9e-3)
    (0.463,8e-4)
};
\addplot [color=green,solid,mark=star,mark size=2.5,mark options={solid},thick]
coordinates{
    (0.448,5e-2)
    (0.446,1.7e-3)
    (0.445,7.4e-5)
};
\addplot [color=green,dashed,mark=star,mark size=2.5,mark options={solid},thick] 
coordinates{
    (0.448,3.1e-1)
    (0.446,6.7e-3)
    (0.445,3e-4)
};
\end{semilogyaxis}
\end{tikzpicture}
\caption{Rate-1/2}
\label{fig:bec1by2}
\end{subfigure}\hfill
\begin{subfigure}[b]{0.47\textwidth}
  \centering
\begin{tikzpicture}
\begin{semilogyaxis}[width=\textwidth,xmin=0.24,xmax=0.32,ymin=1e-5,ymax=1,grid=both,minor grid style=dotted,major grid style=dashed,xlabel=$\epsilon$,legend style={at={(axis cs:0.245,2e-5)},anchor=south west}]
\addplot [color=blue,solid,mark=o,mark size=2.5,mark options={solid},thick]
coordinates {
    (0.318,0.11)
    (0.315,0.015)
    (0.312,0.0009)
    (0.31,0.00003)
};  
\addplot [color=blue,dashed,mark=o,mark size=2.5,mark options={solid},thick]
coordinates {
    (0.318,0.639)
    (0.315,0.086)
    (0.312,0.005)
    (0.31,0.0006) 
};
\addplot [color=black,solid,mark=square,mark size=2.5,mark options={solid},thick]
coordinates {
    (0.289,6e-2)
    (0.288,3e-2)
    (0.286,5e-3)
    (0.284,8.2e-5)
};  
\addplot [color=black,dashed,mark=square,mark size=2.5,mark options={solid},thick]
coordinates {
    (0.289,2.4e-1)
    (0.288,1.2e-1)
    (0.286,2e-2)
    (0.284,3.4e-4)
};
\addplot [color=red,solid,mark=diamond,mark size=2.5,mark options={solid},thick]
coordinates{
    (0.305,6e-2)
    (0.303,8e-3)
    (0.3,2.96e-5)
    (.299,4.66e-7)
};
\addplot [color=red,dashed,mark=diamond,mark size=2.5,mark options={solid},thick] 
coordinates{
    (0.305,5e-1)
    (0.303,6.9e-2)
    (0.3,1.1e-3)
    (.299,3.4e-4)
};
\addplot [color=green,solid,mark=star,mark size=2.5,mark options={solid},thick]
coordinates{
    (0.291,4.2e-2)
    (0.289,1.8e-3)
    (0.288,1.66e-6)
};
\addplot [color=green,dashed,mark=star,mark size=2.5,mark options={solid},thick] 
coordinates{
    (0.291,2.5e-1)
    (0.289,1e-2)
    (0.288,8.6e-4)
};
\legend{Proposed,,AR4JA,,DVB-S2,,WIMAX}
\end{semilogyaxis}
\end{tikzpicture}
\caption{Rate-2/3}
\label{fig:bec2by3}
\end{subfigure}\hfill
\begin{subfigure}[b]{0.47\textwidth}
  \centering
\begin{tikzpicture}
\begin{semilogyaxis}[width=\textwidth,xmin=0.2,xmax=0.24,ymin=1e-5,ymax=1,grid=both,minor grid style=dotted,major grid style=dashed,xlabel=$\epsilon$,ylabel=Error Rates,legend style={at={(axis cs:0.24,2e-5)},anchor=south west}]
\addplot [color=blue,solid,mark=o,mark size=2.5,mark options={solid},thick]
coordinates{
    (0.234,0.025)
    (0.232,0.0045)
    (0.23,0.00029)
    (0.229,0.00006)
};
\addplot [color=blue,dashed,mark=o,mark size=2.5,mark options={solid},thick] 
coordinates{
    (0.234,0.23)
    (0.232,0.086)
    (0.23,0.005)
    (0.229,0.0008)
};
\addplot [color=red,solid,mark=diamond,mark size=2.5,mark options={solid},thick]
coordinates{
    (0.231,2.6e-1)
    (0.23,7e-3)
    (0.228,5.5e-4)
};
\addplot [color=red,dashed,mark=diamond,mark size=2.5,mark options={solid},thick] 
coordinates{
    (0.231,2.9e-1)
    (0.23,9e-2)
    (0.228,1.4e-3)
};
\addplot [color=green,solid,mark=star,mark size=2.5,mark options={solid},thick]
coordinates{
    (0.211,1.6e-2)
    (0.21,2.7e-3)
    (0.209,1.87e-4)
};
\addplot [color=green,dashed,mark=star,mark size=2.5,mark options={solid},thick] 
coordinates{
    (0.211,1.2e-1)
    (0.21,2e-2)
    (0.209,1e-3)
};
\end{semilogyaxis}
\end{tikzpicture}
\caption{Rate-3/4}
\label{fig:bec3by4}
\end{subfigure}
  \caption{Error rates over a BEC. Solid: bit-error
    rate, Dashed: Block-error rate.}
  \label{fig:bec}
\end{figure*}
As seen from the figure, the optimized protographs perform better than
other comparable codes.
\subsection{BIAWGN channel}
The protographs in in Table \ref{tab:optprot_AWGN} are lifted using $D(m,q)$ graphs to
obtained protograph LDPC codes with parameters given in Table \ref{tab:parameter_AWGN}.
\begin{table}[htb]
  \centering
  \begin{tabular}{|c|c|c|c|c|}
    \hline
    {\bf Code type}&{\bf Rate}&$\mathbf{m,q}$&{\bf Blocklength}\\
    \hline
     $16\times 32$ protograph in \eqref{th=0.3}&0.5&2, 173&957728\\
     \hline
     $4\times 12$ protograph in \eqref{th=1.18}&0.66&2, 67&53868\\
     \hline
     $3\times 12$ protograph in \eqref{th=1.79}&0.75&2, 71&60492\\
     \hline
  \end{tabular}
  \caption{Parameters of protograph LDPC codes used for simulation over BIAWGN channel.}
  \label{tab:parameter_AWGN}
\end{table}
The standard message-passing decoder is simulated over the BIAWGN
channel The bit and block-error rates are compared with codes such as
AR4JA and those from DVB-S2 and WIMAX standard in Fig. \ref{fig:awgn}.
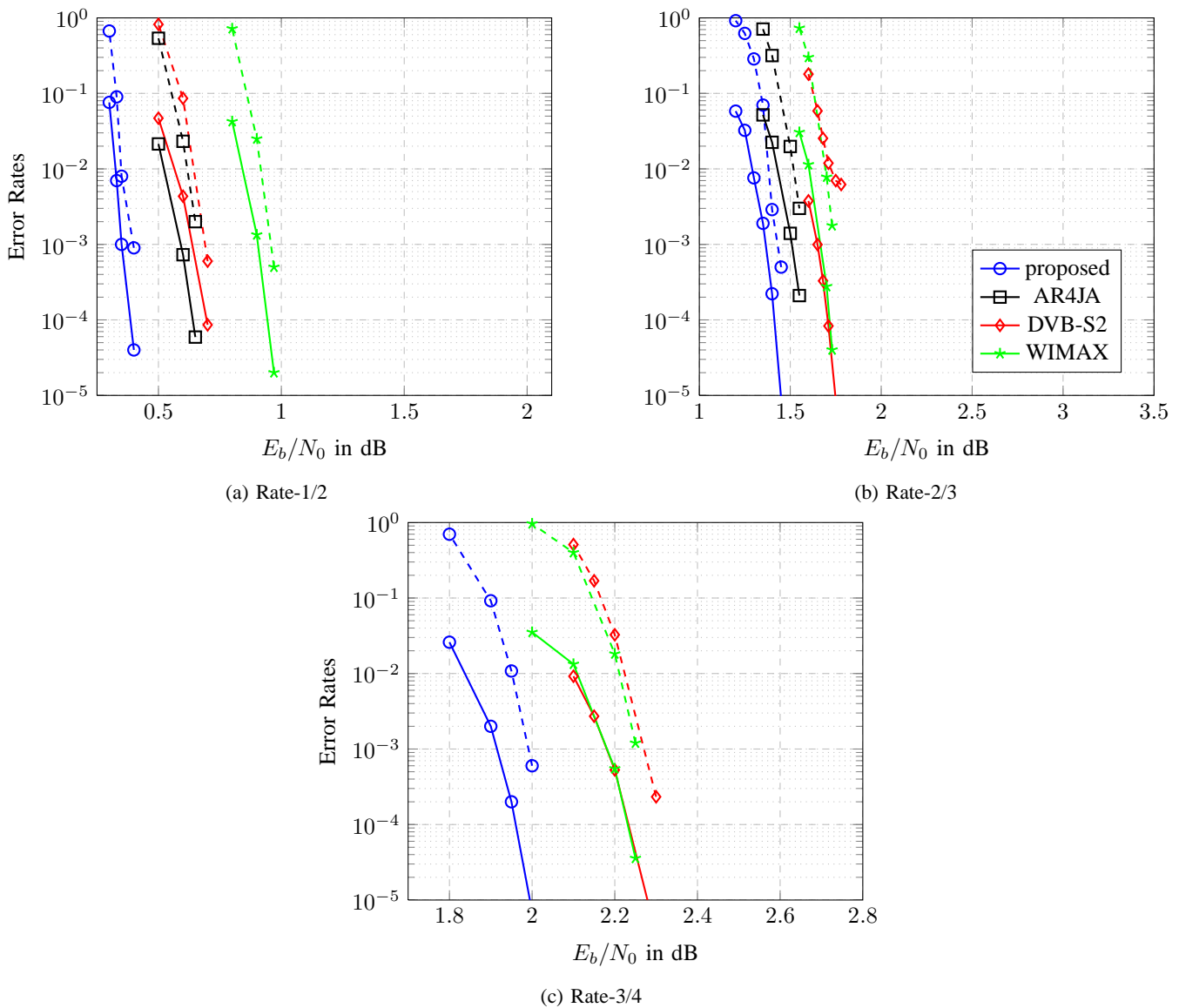
\begin{figure*}
  \centering
  \begin{subfigure}[b]{0.47\textwidth}
   \centering
  \begin{tikzpicture}
\begin{semilogyaxis}[width=\textwidth,xmin=0.25,xmax=2.1,ymin=1e-5,ymax=1,grid=both,minor grid style=dotted,major grid style=dashed,xlabel=$E_b/N_0$ in dB,ylabel=Error Rates,legend style={at={(axis cs:0.6,2e-5)},anchor=south west}]
\addplot [color=blue,solid,mark=o,mark size=2.5,mark options={solid},thick] 
coordinates {
    (0.40,4e-5)
    (0.35,1e-3)
    (0.33,7e-3)
    (0.3,7.6e-2)
};
\addplot [color=blue,dashed,mark=o,mark size=2.5,mark options={solid},thick] 
coordinates {
    (0.40,9e-4)
    (0.35,8e-3)
    (0.33,9e-2)
    (0.3,6.7e-1)
};
\addplot [color=black,solid,mark=square,mark size=2.5,mark options={solid},thick]
coordinates {
    (0.5,2.14e-2)
    (0.6,7.3e-4)
    (0.65,5.92e-5)
};  
\addplot [color=black,dashed,mark=square,mark size=2.5,mark options={solid},thick]
coordinates {
    (0.5,5.37e-1)
    (0.6,2.32e-2)
    (0.65,2e-3)
};
\addplot [color=red,solid,mark=diamond,mark size=2.5,mark options={solid},thick]
coordinates{
    (0.7,8.6e-5)
    (0.6,4.34e-3)
    (0.5,4.69e-2)
};
\addplot [color=red,dashed,mark=diamond,mark size=2.5,mark options={solid},thick]
coordinates{
    (0.7,6e-4)
    (0.6,8.6e-2)
    (0.5,8.21e-1)
};
\addplot [color=green,solid,mark=star,mark size=2.5,mark options={solid},thick]
coordinates{
    (0.97,2e-5)
    (0.9,1.34e-3)
    (0.8,4.24e-2)
};
\addplot [color=green,dashed,mark=star,mark size=2.5,mark options={solid},thick]
coordinates{
    (0.97,5e-4)
    (0.9,2.5e-2)
    (0.8,7.2e-1)
};
\end{semilogyaxis}
\end{tikzpicture}
\caption{Rate-1/2}
\label{rate0.5_AWGN}
\end{subfigure}\hfill
\begin{subfigure}[b]{0.47\textwidth}
   \centering
  \begin{tikzpicture}
\begin{semilogyaxis}[width=\textwidth,xmin=1,xmax=3.5,ymin=1e-5,ymax=1,grid=both,minor grid style=dotted,major grid style=dashed,xlabel=$E_b/N_0$ in dB,legend style={at={(axis cs:2.5,2e-5)},anchor=south west}]
\addplot [color=blue,solid,mark=o,mark size=2.5,mark options={solid},thick] 
coordinates {
    (1.45,9e-6)
    (1.4,2.22e-4)
    (1.35,1.9e-3)
    (1.30,7.6e-3)
    (1.25,3.24e-2)
    (1.2,5.82e-2)
};  
\addplot [color=blue,dashed,mark=o,mark size=2.5,mark options={solid},thick]
coordinates {
    (1.45,5e-4)
    (1.4,2.89e-3)
    (1.35,7e-2)
    (1.30,2.86e-1)
    (1.25,6.23e-1)
    (1.2,9.18e-1)
};
\addplot [color=black,solid,mark=square,mark size=2.5,mark options={solid},thick]
coordinates {
    (1.35,5.19e-2)
    (1.4,2.25e-2)
    (1.5,1.40e-3)
    (1.55,2.1e-4)
};  
\addplot [color=black,dashed,mark=square,mark size=2.5,mark options={solid},thick]
coordinates {
    (1.35,7.11e-1)
    (1.4,3.17e-1)
    (1.5,1.98e-2)
    (1.55,3e-3)
};
\addplot [color=red,solid,mark=diamond,mark size=2.5,mark options={solid},thick] 
coordinates {
    (1.78,7.11e-7)
    (1.75,8.76e-6)
    (1.71,8.31e-5)
    (1.68,3.29e-4)
    (1.65,9.94e-4)
    (1.6,3.77e-3)
};  
\addplot [color=red,dashed,mark=diamond,mark size=2.5,mark options={solid},thick]
coordinates {
    (1.78,6.21e-3)
    (1.75,7e-3)
    (1.71,1.19e-2)
    (1.68,2.54e-2)
    (1.65,5.86e-2)
    (1.6,1.8e-1)
};
\addplot [color=green,solid,mark=star,mark size=2.5,mark options={solid},thick]
coordinates{
    (1.73,4.01e-5)
    (1.7,2.77e-4)
    (1.6,1.14e-2)
    (1.55,3.05e-2)
};
\addplot [color=green,dashed,mark=star,mark size=2.5,mark options={solid},thick]
coordinates{
    (1.73,1.77e-3)
    (1.7,7.8e-3)
    (1.6,3e-1)
    (1.55,7.32e-1)
};

\legend{proposed,,AR4JA,,DVB-S2,,WIMAX,}
\end{semilogyaxis}
\end{tikzpicture}
\caption{Rate-2/3}
\label{rate0.6_AWGN}
\end{subfigure}\hfill
\begin{subfigure}[b]{0.47\textwidth}
 \centering
 \begin{tikzpicture}
\begin{semilogyaxis}[width=\textwidth,xmin=1.7,xmax=2.8,ymin=1e-5,ymax=1,grid=both,minor grid style=dotted,major grid style=dashed,xlabel=$E_b/N_0$ in dB,ylabel=Error Rates,legend style={at={(axis cs:2.32,2e-5)},anchor=south west}]
\addplot [color=blue,solid,mark=o,mark size=2.5,mark options={solid},thick]
coordinates{
    (2.0,7e-6)
    (1.95,2e-4)
    (1.9,2e-3)
    (1.8,2.6e-2)
};
\addplot [color=blue,dashed,mark=o,mark size=2.5,mark options={solid},thick]
coordinates{
    (2,6e-4)
    (1.95,1.08e-2)
    (1.9,9.2e-2)
    (1.8,7e-1)
};
\addplot [color=red,solid,mark=diamond,mark size=2.5,mark options={solid},thick]
coordinates{
    (2.3,3.5e-6)
    (2.2,5.28e-4)
    (2.15,2.72e-3)
    (2.1,9.2e-3)
};
\addplot [color=red,dashed,mark=diamond,mark size=2.5,mark options={solid},thick]
coordinates{
    (2.3,2.33e-4)
    (2.2,3.26e-2)
    (2.15,1.69e-1)
    (2.1,5.1e-1)
};
\addplot [color=green,solid,mark=star,mark size=2.5,mark options={solid},thick]
coordinates{
    (2.25,3.57e-5)
    (2.2,5.49e-4)
    (2.1,1.33e-2)
    (2,3.52e-2)
};
\addplot [color=green,dashed,mark=star,mark size=2.5,mark options={solid},thick]
coordinates{
    (2.25,1.2e-3)
    (2.2,1.82e-2)
    (2.1,4e-1)
    (2,9.62e-1)
};
\end{semilogyaxis}
\end{tikzpicture}
\caption{Rate-3/4}
\label{fig:awgn_0.75}
\end{subfigure}
  \caption{Error rates over a BIAWGN channel.  Solid: bit-error
    rate, Dashed: Block-error rate.}
  \label{fig:awgn}
\end{figure*}
As seen from the figure, the optimized protographs perform better than
other comparable codes. Further, we note that the waterfall region for both
block-error and bit-error rate curves are the same for both the BEC and BIAWGN channels.
\section{Conclusion}
\label{sec:conclusion}
In this work, we studied protograph density evolution and derived
conditions under which the bit and block-error thresholds
coincide. Using large-girth graphs, we presented a deterministic
construction for a sequence of LDPC codes with block-error rate
falling faster than any inverse polynomial in blocklength. We described methods
to optimize protographs and presented small-sized protographs with
thresholds close to capacity.

As part of future work, characterizing the gap to capacity of
finite-length protographs theoretically appears to be an interesting problem for
study, particularly because the thresholds can be close to capacity.  

\bibliographystyle{IEEEtran}
\bibliography{references}

\end{document}